\documentclass[10pt]{article}
\usepackage[english]{babel}								
\usepackage[utf8]{inputenc}								
\usepackage[T1]{fontenc}								
\usepackage{amsmath,amsfonts,amssymb,amsthm,cancel,siunitx,
calculator,calc,mathtools,empheq,latexsym}
\usepackage{subfig,epsfig,tikz,float}		            
\usepackage{pgfplots}                                   
\pgfplotsset{compat=1.18}
\usetikzlibrary{decorations.pathreplacing}              
\usepackage{booktabs,multicol,multirow,tabularx,array}          
\usepackage{natbib}
\usepackage{algorithm}
\usepackage{algpseudocode}
\usepackage{url}                                         
\newtheorem{definition}{Definition}[section]

\newtheorem{theorem}{Theorem}[section]
\newtheorem{proposition}{Proposition}[section]

\title{\normalsize\bf%
\uppercase{SAT Encodings for Bandwidth Coloring: A Systematic Design Study}
}
\author{%
Duc Trung Kim Nguyen$^{1}$, Tuyen Van Kieu$^{2}$, and Khanh Van To$^{3*}$
}

\begin{document}

\date{}

\maketitle

\vspace{-0.5cm}

\begin{center}
{\footnotesize 
*Corresponding author\\
$^1$VNU University of Engineering and Technology, Hanoi, Vietnam / 23021533@vnu.edu.vn / https://orcid.org/0009-0001-7828-4985\\
$^2$VNU University of Engineering and Technology, Hanoi, Vietnam / tuyenkv@vnu.edu.vn / https://orcid.org/0009-0007-7800-7172\\
$^3$VNU University of Engineering and Technology, Hanoi, Vietnam / khanhtv@vnu.edu.vn / https://orcid.org/0009-0008-1907-7848\\
 }
\end{center}

\bigskip
\noindent
{\small{\bf ABSTRACT.}
The Bandwidth Coloring Problem (BCP) generalizes graph coloring by enforcing minimum separation constraints between adjacent vertices and arises in frequency assignment applications. While SAT-based approaches have shown promise for exact BCP solving, the encoding design space remains largely unexplored.
This paper presents a systematic study of SAT encodings for the BCP, proposing a unified framework with six encoding methods across three categories: one-variable, two-variable, and block encodings. We evaluate the impact of key features including incremental solving and symmetry breaking. While symmetry breaking has been studied for graph coloring, it has not been systematically evaluated for SAT-based BCP solvers. Our analysis reveals significant interaction effects between encoding choices and solver configurations.
The proposed framework achieves state-of-the-art performance on GEOM and MS-CAP benchmarks. Block encodings solve GEOM120b, the hardest instance, to proven optimality in approximately 1000 seconds, whereas previous methods could not solve it within a one-hour time limit.
}


\baselineskip=\normalbaselineskip

\section{Introduction}\label{sec:intro}

The graph coloring problem (GCP) is a classical combinatorial optimization problem that asks for an assignment of colors to the vertices of a graph such that no two adjacent vertices share the same color while minimizing the number of colors used~\citep{malaguti2010,husfeldt2015}. This problem has extensive applications in register allocation~\citep{chaitin1981}, scheduling~\citep{leighton1979}, and computing sparse Jacobian matrices~\citep{gebremedhin2005}. Despite its theoretical and practical importance, finding an optimal coloring remains NP-hard, and compared to other well-known NP-hard problems such as the traveling salesman problem, only relatively small instances can be solved to provable optimality~\citep{faber2024}.

The bandwidth coloring problem (BCP) is a natural generalization of the GCP in which each edge $\{u, v\}$ carries a positive integer weight $d(u, v)$ representing a minimum required separation between the colors assigned to its endpoints. 
Formally, a coloring is defined by a labeling function $c : V \rightarrow \mathbb{N}$, where $c(u)$ and $c(v)$ denote the colors assigned to vertices $u$ and $v$, respectively. A coloring is feasible if it satisfies the distance constraint $|c(u) - c(v)| \geq d(u,v)$ for every edge $\{u,v\} \in E$. The objective is to minimize the largest color assigned to any vertex, commonly referred to as the \emph{span}~\citep{marti2010}. 
When all edge weights equal one, the BCP reduces to the classical GCP. The BCP arises naturally in frequency assignment problems (FAP) for wireless and cellular networks~\citep{hale1980,aardal2007}. In these applications, radio transmitters are modeled as vertices, and edge weights represent minimum frequency separations required to avoid electromagnetic interference~\citep{eisenblatter2002,koster1999}. The span of the coloring corresponds to the bandwidth required for the network, making minimization crucial for efficient spectrum utilization~\citep{dias2021,lai2013}.

The predominant approaches for solving the BCP can be broadly classified into two categories: metaheuristic methods and exact algorithms~\citep{marti2010}. Metaheuristic approaches, including multistart iterated tabu search~\citep{lai2013}, learning-based hybrid search~\citep{jin2015}, path relinking~\citep{lai2014,lai2016}, and variable neighborhood search~\citep{matic2017}, have demonstrated fast computation times but provide no guarantee of optimality. In contrast, exact methods based on integer linear programming (ILP) and constraint programming (CP) have been investigated by~\citet{dias2021}, who observed that CP formulations are particularly effective for BCP due to efficient constraint propagation mechanisms, while ILP formulations tend to suffer from a large number of constraints that grows with edge weights.

More recently, SAT-based approaches have emerged as a promising alternative for solving graph coloring variants exactly~\citep{hebrard2020,heule2022}. The assignment-based SAT encoding introduces binary variables to represent color assignments directly~\citep{heule2022,dey2013}, while order-based encodings use variables indicating relative ordering between colors and vertices~\citep{tamura2009,ansotegui2004}. For the GCP, \citet{heule2022} proposed the CliColCom algorithm that alternates between maximum clique and graph coloring problems using SAT approaches. For the BCP, \citet{faber2024} proposed partial-ordering based SAT encodings (POP-S-B and POPH-S-B), which represent color assignments through ordering variables rather than direct assignment variables. Their formulations achieve an asymptotically smaller number of constraints compared to assignment-based models and have successfully solved six previously open GEOM benchmark instances to optimality. However, the design space for SAT encodings remains largely unexplored, with several important questions unanswered. First, although the POP-S-B and POPH-S-B encodings have demonstrated promising results, a comprehensive evaluation of alternative variable representations, including one-variable versus two-variable encodings and greater-than versus less-than semantics, has not been conducted. Second, while symmetry breaking strategies have been studied extensively for the GCP~\citep{mendezdiaz2008}, their effectiveness for the BCP remains uninvestigated. Third, the impact of incremental SAT solving versus non-incremental approaches on BCP performance remains unclear~\citep{glorian2019}. Fourth, block encodings~\citep{truong2025sequential}, which organize variables into structured groups for more efficient constraint propagation, have not been applied to the BCP. Fifth, the interaction effects between these configurable features, specifically whether they operate independently or exhibit synergistic or antagonistic relationships, remain unknown.

This paper addresses these gaps by presenting a systematic design study of SAT encodings for the BCP. We propose a comprehensive framework that explores the design choices for SAT-based BCP solvers, encompassing six encoding methods organized into three categories: one-variable encodings (1G and 1L), two-variable encodings (2G and 2L), and block encodings (X and Xa). For each encoding type, we evaluate the impact of key configurable features: incremental solving modes, symmetry breaking, and block width strategies.

The main contributions of this paper are as follows. First, we demonstrate that the proposed block encodings are able to solve GEOM120b, the most challenging benchmark instance, to proven optimality in approximately 1000 seconds, whereas previous approaches could not solve this instance within a one-hour time limit.
Second, we present a unified SAT encoding framework for the BCP and conduct a comprehensive comparative study that identifies factors influencing solver performance. Third, we observe important interaction effects: incremental solving benefits block encodings substantially but not one-variable encodings, and symmetry breaking effectiveness varies significantly across encoding types.

The remainder of this paper is organized as follows. Section~\ref{sec:problem} formally defines the bandwidth coloring problem and establishes the notation used throughout. Section~\ref{sec:literature} reviews related work on metaheuristic and exact approaches for the BCP. Section~\ref{sec:method} presents the encoding framework in detail, describing the six encoding methods and their configurable features. Section~\ref{sec:experiments} reports computational experiments on GEOM and MS-CAP benchmark instances. Finally, Section~\ref{sec:conclusions} summarizes the findings and discusses directions for future research.


\section{Problem Definition}\label{sec:problem}

This section formally defines the bandwidth coloring problem and related variants, and establishes the notation used throughout the paper.

\subsection{Graph Coloring Problem}

Let $G = (V, E)$ be an undirected graph with vertex set $V$ and edge set $E$, where each edge is a two-element subset $e = \{u, v\}$ of $V$. The vertices $u$ and $v$ of an edge $\{u, v\}$ are called \emph{adjacent vertices} or \emph{neighbors}. A \emph{$k$-coloring} of $G$ is an assignment of colors from the set $\{1, 2, \ldots, k\}$ to the vertices such that no two adjacent vertices receive the same color~\citep{malaguti2010}. The \emph{chromatic number} $\chi(G)$ is the minimum value of $k$ for which a valid $k$-coloring exists.

\begin{definition}[Graph Coloring Problem]
Given an undirected graph $G = (V, E)$, the graph coloring problem (GCP) asks for an assignment $c: V \rightarrow \mathbb{N}$ that minimizes $\max_{v \in V} c(v)$ subject to the constraint $c(u) \neq c(v)$ for all $\{u, v\} \in E$.
\end{definition}

Table~\ref{table:notation} summarizes the notation used throughout this paper.

\begin{table}[!htb]
\centering
\caption{Notation used in this paper.}
\begin{tabularx}{0.95\textwidth}{l X}
\toprule
\multicolumn{1}{c}{\textbf{Symbol}} & \multicolumn{1}{c}{\textbf{Description}} \\
\midrule
$G = (V, E)$ & Undirected graph with vertex set $V$ and edge set $E$ \\
$n = |V|$ & Number of vertices \\
$m = |E|$ & Number of edges \\
$d(\{u,v\})$ or $d_{u,v}$ & Distance requirement (weight) for edge $\{u, v\}$ \\
$\bar{d}$ & Average edge distance: $\bar{d} = \frac{1}{|E|} \sum_{e \in E} d(e)$ \\
$k$ & Span (largest color used in a coloring) \\
$H$ & Upper bound on the span \\
$c(v)$ & Color assigned to vertex $v$ \\
$\chi(G)$ & Chromatic number of graph $G$ \\
$w(v)$ & Color demand for vertex $v$ (in BMCP) \\
\bottomrule
\end{tabularx}
\label{table:notation}
\end{table}

\subsection{Bandwidth Coloring Problem}

The bandwidth coloring problem generalizes the GCP by introducing edge weights that specify minimum separation requirements between colors of adjacent vertices~\citep{marti2010}. This generalization captures the practical constraint in frequency assignment where nearby transmitters require sufficiently different frequencies to avoid interference~\citep{dias2021}.

\begin{definition}[Bandwidth Coloring Problem]
Given an undirected graph $G = (V, E)$ with positive integer edge weights $d: E \rightarrow \mathbb{N}$, the bandwidth coloring problem (BCP) asks for an assignment $c: V \rightarrow \mathbb{N}$ that minimizes $\max_{v \in V} c(v)$ subject to the constraint:
\begin{equation}\label{eq:bcp-constraint}
|c(u) - c(v)| \geq d(\{u, v\}) \quad \forall \{u, v\} \in E
\end{equation}
\end{definition}

Figure~\ref{fig:bcp-example} illustrates the BCP with a small example. The graph has four vertices connected by edges with weights 1, 2, and 3. A valid coloring assigns colors such that adjacent vertices satisfy the distance constraint. For instance, edge $\{A, B\}$ has weight 2, so $|c(A) - c(B)| = |1 - 3| = 2 \geq 2$ is satisfied. The right panel shows the constraint verification for all five edges, confirming the validity of the coloring with span 6.

\begin{figure}[htb]
\centering
\begin{tikzpicture}[
    vertex/.style={circle, draw, minimum size=8mm, font=\small},
    edge label/.style={font=\small, fill=white, inner sep=1pt}
]
\node[anchor=north] at (1.5, 3.5) {\textbf{(a) Graph with edge weights}};

\node[vertex] (A) at (0, 2) {$A$};
\node[vertex] (B) at (3, 2) {$B$};
\node[vertex] (C) at (3, 0) {$C$};
\node[vertex] (D) at (0, 0) {$D$};

\draw[thick] (A) -- node[edge label, above] {$d=2$} (B);
\draw[thick] (B) -- node[edge label, right] {$d=3$} (C);
\draw[thick] (C) -- node[edge label, below] {$d=1$} (D);
\draw[thick] (D) -- node[edge label, left] {$d=2$} (A);
\draw[thick] (A) -- node[edge label, pos=0.3, above] {$d=1$} (C);

\node[anchor=north] at (7.5, 3.5) {\textbf{(b) Valid coloring (span = 6)}};

\node[vertex, fill=blue!30] (A2) at (6, 2) {$A$};
\node[vertex, fill=green!30] (B2) at (9, 2) {$B$};
\node[vertex, fill=yellow!50] (C2) at (9, 0) {$C$};
\node[vertex, fill=red!30] (D2) at (6, 0) {$D$};

\draw[thick] (A2) -- (B2);
\draw[thick] (B2) -- (C2);
\draw[thick] (C2) -- (D2);
\draw[thick] (D2) -- (A2);
\draw[thick] (A2) -- (C2);

\node[font=\small] at (6, 2.7) {$c=1$};
\node[font=\small] at (9, 2.7) {$c=3$};
\node[font=\small] at (9, -0.7) {$c=6$};
\node[font=\small] at (6, -0.7) {$c=4$};

\node[anchor=west, font=\footnotesize] at (9.8, 2.5) {$|1-3|=2 \geq 2$ \checkmark};
\node[anchor=west, font=\footnotesize] at (9.8, 1.8) {$|3-6|=3 \geq 3$ \checkmark};
\node[anchor=west, font=\footnotesize] at (9.8, 1.1) {$|6-4|=2 \geq 1$ \checkmark};
\node[anchor=west, font=\footnotesize] at (9.8, 0.4) {$|4-1|=3 \geq 2$ \checkmark};
\node[anchor=west, font=\footnotesize] at (9.8, -0.3) {$|1-6|=5 \geq 1$ \checkmark};

\end{tikzpicture}
\caption{Example of the Bandwidth Coloring Problem.}
\label{fig:bcp-example}
\end{figure}
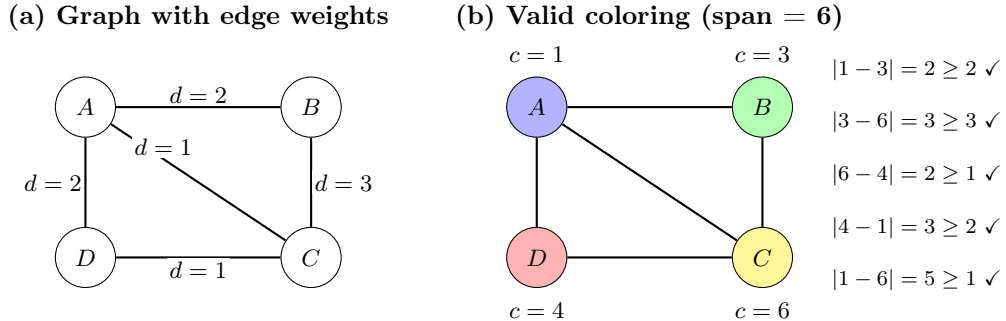

The objective value $\max_{v \in V} c(v)$ is called the \emph{span} of the coloring. When all edge weights are equal to one (i.e., $d(e) = 1$ for all $e \in E$), the BCP reduces to the classical GCP~\citep{faber2024}.

The BCP is a special case of the T-coloring problem~\citep{cozzens1982,roberts1991}, where each edge is associated with a forbidden set of color differences. A further generalization is the bandwidth multicoloring problem (BMCP), which allows each vertex to receive multiple colors~\citep{lai2013,dias2021}. Standard transformations allow a BMCP instance to be converted into an equivalent BCP instance by replacing each vertex with a clique that represents its color demand, thereby enabling BCP algorithms to be applied to problems originally formulated as BMCP.
The GCP is known to be NP-hard~\citep{malaguti2010}, and since the BCP generalizes the GCP, the BCP is also NP-hard. The decision version of the problem, which asks whether a valid $k$-coloring exists for a given $k$, is NP-complete.


\section{Literature Review}\label{sec:literature}

This section reviews existing approaches for solving the bandwidth coloring problem, organized into three categories: metaheuristic methods, integer and constraint programming approaches, and SAT-based methods. We conclude with an analysis of research gaps that motivate the present work. While we provide a comprehensive overview of all approaches for context, our experimental comparison in Section~\ref{sec:experiments} focuses exclusively on SAT-based methods, as they represent the current state of the art for exact BCP solving~\citep{faber2024}.

\subsection{Metaheuristic Approaches}

Metaheuristic algorithms have been extensively studied for the BCP due to their ability to find high-quality solutions in reasonable time. Notable approaches include Multistart Iterated Tabu Search (MITS)~\citep{lai2013}, which employs fast incremental evaluation techniques to efficiently update the objective function during neighborhood exploration; Learning-based Hybrid Search (LHS)~\citep{jin2015}, which combines forward checking construction with tabu search repair using adaptive guiding functions; path relinking strategies~\citep{lai2014,lai2016} that balance intensification and diversification through randomized trajectory generation; and Variable Neighborhood Search (VNS)~\citep{matic2017}, which decomposes local search into discrete procedures based on conflict counts and edge weight characteristics.

Despite their computational efficiency, metaheuristic approaches share fundamental limitations. They provide no guarantee of solution optimality, which is critical in applications such as frequency assignment where suboptimal allocations may lead to interference or wasted spectrum. These methods are also susceptible to becoming trapped in local optima, and their performance often depends heavily on parameter tuning specific to problem instances. These limitations motivate the development of exact methods capable of proving optimality.

\subsection{Integer and Constraint Programming}

Exact methods based on mathematical programming have been investigated to prove optimality for BCP instances. \citet{dias2021} conducted a comparative study of integer programming (IP) and constraint programming (CP) formulations. The CP formulation directly uses distance constraints $|c(u) - c(v)| \geq d_{u,v}$, which can be handled naturally by constraint propagation mechanisms. The assignment-based ILP model (ASS-I-B), as presented by~\citet{faber2024}, uses binary variables $x_{v, i}$ (where $x_{v,i} = 1$ iff vertex $v$ is assigned color $i$) and a variable $z_{\max}$ to represent the span. The formulation is given by:

\begin{align}
\min \quad & z_{\max} \\
\text{s.t.} \quad & \sum_{i=1}^{H} x_{v, i} = 1 && \forall v \in V \\
& x_{u, i} + x_{v, j} \leq 1 && \forall \{u, v\} \in E, \forall i, j: |i - j| < d(u, v) \\
& z_{\max} \geq i \cdot x_{v, i} && \forall v \in V, \forall i \in \{1, \ldots, H\} \\
& x_{v, i} \in \{0, 1\}, z_{\max} \in \mathbb{Z}_{\geq 0} && \forall v \in V, \forall i \in \{1, \ldots, H\}
\end{align}

This model ensures each vertex receives exactly one color, prevents conflicting color assignments for adjacent vertices, and correctly defines the objective function. However, both IP and CP formulations exhibit significant limitations for the BCP. The assignment-based IP model suffers from two major drawbacks: (1) the presence of symmetries in the solution space, where any permutation of color labels produces an equivalent solution, leading to a significantly larger search space; and (2) the number of constraints grows with $O(H \cdot |E| \cdot \bar{d})$, where $\bar{d}$ is the average edge weight, making the formulation impractical for instances with large edge weights~\citep{faber2024}. While \citet{dias2021} observed that CP is effective for BCP instances due to efficient propagation of distance constraints, both approaches struggle with scalability on larger instances. These limitations have motivated the exploration of SAT-based approaches, which offer alternative constraint representations with potentially more favorable scaling properties.

\subsection{SAT-based Approaches}

SAT-based methods have emerged as a competitive alternative for exact graph coloring. The \emph{assignment-based} SAT encoding uses binary variables $x_{v,i}$, where $x_{v,i} = \text{true}$ iff vertex $v$ is assigned color $i$, to represent color assignments directly~\citep{heule2022,dey2013}. Constraints ensure that each vertex receives exactly one color and that adjacent vertices receive different colors.

\citet{tamura2009} and \citet{ansotegui2004} proposed an \emph{order-based} encoding using variables $y_{v,i}$ that indicate whether color $i$ is smaller than the color assigned to vertex $v$. This encoding has theoretical advantages in constraint count but has received limited attention in subsequent literature until recently.

For the GCP, \citet{heule2022} introduced the CliColCom algorithm, which alternates between solving maximum clique and graph coloring problems using SAT, where solutions from one problem guide the search for the other. \citet{hebrard2020} developed constraint and satisfiability reasoning techniques that incorporate clause learning specifically designed for graph coloring. \citet{glorian2019} proposed an incremental SAT approach that adds constraints dynamically during the search.

For the BCP specifically, \citet{faber2024} proposed partial-ordering based SAT encodings (POP-S-B and POPH-S-B) that generalize the ordering approach of \citet{tamura2009}. Their key insight is that the number of constraints depends only on $H \cdot |E|$ rather than $H \cdot |E| \cdot \bar{d}$ as in assignment-based models. This asymptotic advantage translates to practical performance gains: their encodings solved six previously open GEOM benchmark instances to optimality for the first time. The work also adapted symmetry-breaking constraints from \citet{mendezdiaz2008} to the SAT context.

\subsection{Research Gaps}

Despite the progress in SAT-based methods for the BCP, several important questions remain unexplored. This subsection identifies four key research gaps that motivate the present work.

\begin{itemize}
    \item \textit{Variable representations:} Existing SAT encodings for the BCP use only greater-than semantics ($y_{v,i} = 1$ iff $c(v) > i$)~\citep{faber2024}, and the relative performance of less-than semantics, one-variable versus two-variable representations, and block encodings has not been investigated.    \item \textit{Incremental solving:} While incremental SAT has been explored for the GCP~\citep{glorian2019}, its effectiveness for the BCP with distance constraints remains unclear; the trade-off between clause reuse and learned clause quality deserves systematic study.
    \item \textit{Symmetry breaking:} Current symmetry-breaking approaches focus on color symmetries~\citep{mendezdiaz2008}, but alternative symmetry breaking strategies have not been systematically evaluated for SAT-based BCP solvers.
    \item \textit{Block encoding strategies:} Block encodings with different block width strategies (fixed versus varying) have not been systematically compared.
\end{itemize}

The present work addresses these gaps by proposing a systematic framework that explores encoding design choices including variable representations, symmetry breaking strategies, solving modes, and block encoding variants.


\section{SAT Encoding Framework}\label{sec:method}

This section presents the SAT encoding framework, which systematically explores the design space for BCP solvers. We describe six encoding methods implemented in our framework: one-variable encodings (1G and 1L), two-variable encodings (2G and 2L), and block encodings (X without auxiliary variables and Xa with auxiliary variables). For each encoding, we present the variable definitions, constraint formulations, and implementation details.

\subsection{Framework Overview}

The proposed framework consists of six encoding methods with configurable features, yielding 36 distinct configurations. Each configuration combines an encoding method with optional features: incremental solving modes, symmetry breaking on the highest-degree vertex, and block width strategies (for X and Xa). Table~\ref{table:config-matrix} summarizes the available features for each encoding method.

\begin{table}[!htb]
\centering
\caption{Configuration matrix: features available for each encoding method.}
\begin{tabular}{l c c c c}
\toprule
\textbf{Method} & \textbf{Width} & \textbf{Incremental} & \textbf{Symmetry} & \textbf{\# Configs} \\
\midrule
1G & -- & $y$, none & true, false & 4 \\
1L & -- & $y$, none & true, false & 4 \\
2G & -- & $x$, $y$, none & true, false & 6 \\
2L & -- & $x$, $y$, none & true, false & 6 \\
X & fixed, vary & $x$, none & true, false & 8 \\
Xa & fixed, vary & $x$, none & true, false & 8 \\
\midrule
& & & Total: & 36 \\
\bottomrule
\end{tabular}
\begin{flushleft}
\footnotesize{1G: one-variable greater-than; 1L: one-variable less-than; 2G: two-variable greater-than; 2L: two-variable less-than; X: block encoding without auxiliary variables; Xa: block encoding with auxiliary variables. Width: block width strategy (fixed or varying). Incremental: variable type used for assumptions ($x$, $y$, or none for non-incremental).}
\end{flushleft}
\label{table:config-matrix}
\end{table}

Each configuration is specified by the following parameters:
\begin{itemize}
    \item \textit{method}: Encoding method (1G, 1L, 2G, 2L, X, Xa);
    \item \textit{width}: Block width strategy (fixed-width, vary-width for X and Xa; empty for others);
    \item \textit{incremental}: Variable type for incremental solving assumptions ($x$, $y$, or empty for non-incremental);
    \item \textit{symmetry}: Whether symmetry breaking is enabled (true or false).
\end{itemize}
For example, the configuration (Xa, fixed-width, $x$, true) denotes the block encoding with fixed block width, incremental solving using $x$ variables for assumptions, and symmetry breaking enabled.

\subsection{Upper Bound Computation}

SAT-based approaches require an upper bound $H$ on the span to define the variable domains~\citep{faber2024}. A trivial upper bound for the BCP is $H = |V| \cdot \max_{e \in E} d(e)$, corresponding to the case where each vertex receives a distinct color with maximum separation. In practice, tighter bounds can be obtained using greedy heuristics. A commonly used approach based on DSatur (degree of saturation)~\citep{brelaz1979} processes vertices sequentially, assigning each vertex the smallest feasible color that satisfies all constraints with previously colored neighbors. The resulting coloring provides an upper bound $H$ that is typically much smaller than the trivial bound and directly impacts the size of the SAT encoding. Algorithm~\ref{alg:dsatur} describes this procedure.

\begin{algorithm}[htb]
\caption{DSatur-based Upper Bound Computation}
\label{alg:dsatur}
\begin{algorithmic}[1]
\Require Graph $G = (V, E)$ with edge weights $d$
\Ensure Upper bound $H$ on the optimal span
\State $c(v) \gets $ unassigned for all $v \in V$
\State $H \gets 0$
\While{exists uncolored vertex}
    \State $v \gets$ uncolored vertex with maximum saturation (ties broken by degree)
    \State Compute forbidden intervals: $I \gets \bigcup_{u \in N(v), c(u) \neq \text{unassigned}} [c(u) - d_{v,u} + 1, c(u) + d_{v,u} - 1]$
    \State $c(v) \gets$ smallest positive integer not in $I$
    \State $H \gets \max(H, c(v))$
\EndWhile
\State \Return $H$
\end{algorithmic}
\end{algorithm}

\subsection{One-Variable Encodings}

One-variable encodings represent color assignments using only ordering variables, requiring $kn$ Boolean variables for a graph with $n$ vertices and span $k$. The color of each vertex is determined implicitly by the transition point in the ordering sequence. We define two variants based on the semantics of the ordering variables: Greater-than (1G) and Less-than (1L).

\subsubsection{Greater-than Encoding (1G)}

The 1G encoding uses Boolean variables $y_{u,j}$ for each vertex $u \in V$ and color $j \in [1, k]$ with the semantic $y_{u,j} = \text{true} \Leftrightarrow c(u) \geq j$. This formulation is commonly used in order encodings for CSP and SAT~\citep{tamura2009}.
A vertex $u$ is assigned color $j$ if $y_{u,j}$ is true and $y_{u,j+1}$ is false (with boundary condition $y_{u,k+1} = \text{false}$).

The constraints are defined as follows:
\begin{itemize}
    \item \textit{Base ranges:} $y_{u,1} = \text{true}$ for all $u \in V$.
    \item \textit{Ordering:} $y_{u,j} \rightarrow y_{u,j-1}$ for all $j \in [2, k]$.
    \item \textit{Distance:} For each edge $\{u,v\}$ with weight $d_{u,v}$ and color $j$, if $c(u) = j$, then $c(v)$ must be outside $(j-d_{u,v}, j+d_{u,v})$.
    This implies $\neg (y_{u,j} \wedge \neg y_{u,j+1}) \vee (\neg y_{v, j-d_{u,v}+1} \vee y_{v, j+d_{u,v}})$.
    Rearranging gives the clause:
    \begin{equation}
    y_{u,j} \rightarrow (y_{u,j+1} \vee \neg y_{v, j-d_{u,v}+1} \vee y_{v, j+d_{u,v}})
    \end{equation}
    Boundary conditions apply (e.g., if $j+d_{u,v} > k$, then $y_{v, j+d_{u,v}}$ is false).
\end{itemize}

\subsubsection{Less-than Encoding (1L)}

The 1L encoding employs variables $y_{u,j}$ with the semantic $y_{u,j} = \text{true} \Leftrightarrow c(u) \leq j$.
\textit{Note:} This differs from the encoding presented by \citet{faber2024}, which uses strictly less-than semantics (i.e., ordering variables indicate $c(u) < j$, or equivalently $c(u)$ is smaller than $j$). Our inclusive definition ($c(u) \leq j$) aligns symmetrically with the 1G encoding.

The constraints for 1L are:
\begin{itemize}
    \item \textit{Base ranges:} $y_{u,k} = \text{true}$ for all $u \in V$.
    \item \textit{Ordering:} $y_{u,j} \rightarrow y_{u,j+1}$ for all $j \in [1, k-1]$.
    \item \textit{Distance:} The distance constraint for edge $\{u,v\}$ with weight $d_{u,v}$ enforces that if $c(u) = j$, then $c(v)$ is not in $(j-d_{u,v}, j+d_{u,v})$.
    In terms of $\leq$ variables, $c(v) \leq j-d_{u,v}$ corresponds to $y_{v, j-d_{u,v}}$, and $c(v) \geq j+d_{u,v}$ corresponds to $\neg y_{v, j+d_{u,v}-1}$.
    The clause is:
    \begin{equation}
    y_{u,j} \rightarrow (\neg y_{u,j-1} \vee y_{v, j-d_{u,v}} \vee \neg y_{v, j+d_{u,v}-1})
    \end{equation}
\end{itemize}

\subsection{Proposed Two-Variable Encodings (2G and 2L)}

While the one-variable encodings derive color assignments implicitly from ordering variables, we propose extending this approach by explicitly maintaining both assignment variables $x_{u,j}$ and ordering variables $y_{u,j}$. Although this increases the number of variables, explicitly modeling $x_{u,j}$ allows for more direct and compact distance constraints, potentially improving unit propagation in the SAT solver.

\subsubsection{Two-Variable Greater-than (2G)}

The 2G encoding uses the following variables:
\begin{itemize}
    \item $x_{u,j} = \text{true} \Leftrightarrow c(u) = j$
    \item $y_{u,j} = \text{true} \Leftrightarrow c(u) \geq j$
\end{itemize}

The constraints are defined as follows:
\begin{itemize}
    \item \textit{Channeling:} Link $x$ and $y$ variables.
    \begin{equation}
    x_{u,j} \leftrightarrow (y_{u,j} \wedge \neg y_{u,j+1})
    \end{equation}
    \item \textit{Ordering:} $y_{u,j} \rightarrow y_{u,j-1}$.
    \item \textit{Distance:} Using $x_{u,j}$, the distance constraint simplifies to:
    \begin{equation}
    x_{u,j} \rightarrow (\neg y_{v, j-d_{u,v}+1} \vee y_{v, j+d_{u,v}})
    \end{equation}
    This states directly: if $u$ is assigned color $j$, then $c(v) \leq j-d_{u,v}$ or $c(v) \geq j+d_{u,v}$.
\end{itemize}

\subsubsection{Two-Variable Less-than (2L)}

The 2L encoding uses the following variables:
\begin{itemize}
    \item $x_{u,j} = \text{true} \Leftrightarrow c(u) = j$
    \item $y_{u,j} = \text{true} \Leftrightarrow c(u) \leq j$
\end{itemize}

The constraints are defined as follows:
\begin{itemize}
    \item \textit{Channeling:}
    \begin{equation}
    x_{u,j} \leftrightarrow (y_{u,j} \wedge \neg y_{u,j-1})
    \end{equation}
    \item \textit{Ordering:} $y_{u,j} \rightarrow y_{u,j+1}$.
    \item \textit{Distance:}
    \begin{equation}
    x_{u,j} \rightarrow (y_{v, j-d_{u,v}} \vee \neg y_{v, j+d_{u,v}-1})
    \end{equation}
\end{itemize}

\begin{proposition}[Encoding Size Comparison]\label{prop:encoding-size}
Let $n=|V|$, $m=|E|$, $k$ be the span, and $\bar{d}$ be the average edge weight. Table~\ref{table:encoding-size} compares the asymptotic size of the proposed encodings with those from~\citet{faber2024}.

\begin{table}[!htb]
\centering
\caption{Comparison of encoding sizes.}
\begin{tabular}{l l c c}
\toprule
\textbf{Source} & \textbf{Encoding} & \textbf{Variables} & \textbf{Clauses} \\
\midrule
\multirow{3}{*}{\citet{faber2024}} 
    & POP-S-B & $(k-1)n$ & $O(k(n+m))$ \\
    & POPH-S-B & $2kn$ & $O(k(n+m))$ \\
    & ASS-S-B & $kn$ & $O(km(2\bar{d}-1))$ \\
\midrule
\multirow{2}{*}{Proposed} 
    & 1G / 1L & $(k-1)n$ & $O(k(n+m))$ \\
    & 2G / 2L & $2kn$ & $O(k(n+m))$ \\
\bottomrule
\end{tabular}
\label{table:encoding-size}
\end{table}

\normalfont
The proposed 1G and 1L encodings have the same asymptotic complexity as POP-S-B, while 2G and 2L match POPH-S-B (with a constant factor increase due to channeling constraints). All order-based encodings have clause count independent of edge weights $\bar{d}$, offering asymptotic advantages over ASS-S-B for BCP instances with large weights.
\end{proposition}

\subsection{Block Encodings (X, Xa)}

While the one-variable and two-variable encodings generate distance clauses for each color value, the block encodings take a different approach: they partition the color range into fixed-width blocks and use auxiliary \textit{range variables} to compactly represent color assignments within each block. This reduces the number of clauses for distance constraints, especially when edge weights are large.

Assignment-based encodings suffer from a large number of clauses when edge weights are large ($O(d_{u,v})$ clauses per edge-color pair). The block encoding addresses this by decomposing the distance constraint summation into blocks, reducing the size dependency on $d_{u,v}$. We adapt the block encoding from~\citet{faber2024} and propose two variants based on how auxiliary variables are handled.

\subsubsection{Variable Definitions and Staircase Block Structure}

The span $[1, k]$ is divided into $m$ blocks $B_1, B_2, \ldots, B_m$.
The framework supports two block width strategies:
\begin{itemize}
    \item \textit{Fixed-width:} All blocks have the same width $w$ (default $w=8$), which efficiently captures most distance constraints within a single block or two adjacent blocks.
    \item \textit{Vary-width:} For each vertex $u \in V$, the block width is determined by the maximum edge weight incident to $u$: $w_u = \max_{\{u,v\} \in E} d_{u,v}$.
\end{itemize}

We introduce auxiliary variables $R_{u, s, e}$ representing the sum of assignment variables in a range:
\begin{equation}
R_{u, a, b} = \text{true} \quad \Leftrightarrow \quad \sum_{j=a}^{b} x_{u,j} = 1
\end{equation}
where $a$ and $b$ are the start and end indices of a block or sub-block.

The exactly-one color constraint for vertex $u$ is enforced over the blocks:
\begin{equation}
\sum_{r=1}^{m} R_{u, s(B_r), e(B_r)} = 1
\end{equation}
enforcing that exactly one block contains the assigned color.

Figure~\ref{fig:block-structure} illustrates the staircase block structure for a vertex $u$ with block width $w=8$ and span $k = 8m$ (for some integer $m$). The first block uses a backward chain, middle blocks use bidirectional chains, and the last block uses a forward chain.

\begin{figure*}[h!]
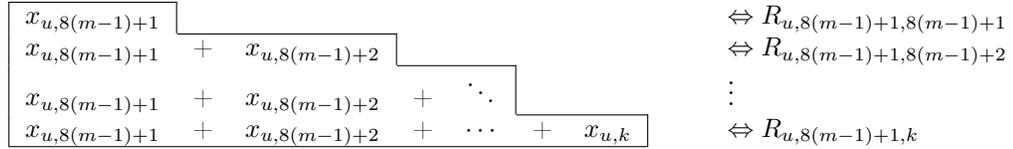

\centering
\resizebox{\textwidth}{!}{
    \begin{tabular}{llllllllllllllllllllllll}
    \cline{1-9}
    \multicolumn{1}{|l}{$x_{u,1}$} & + & $x_{u,2}$ & + & $x_{u,3}$ & + & $\cdots$ & + & \multicolumn{1}{l|}{$x_{u,8}$} &  &  &  &  &  &  &  &  &  &  &  &  &  &  & $\leq 1$ \\ \cline{1-2} \cline{11-11}
     & \multicolumn{1}{l|}{} & $x_{u,2}$ & + & $x_{u,3}$ & + & $\cdots$ & + & \multicolumn{1}{l|}{$x_{u,8}$} & \multicolumn{1}{l|}{+} & \multicolumn{1}{l|}{$x_{u,9}$} &  &  &  &  &  &  &  &  &  &  &  &  & $\leq 1$ \\ \cline{3-4} \cline{12-13}
     &  &  & \multicolumn{1}{l|}{} & $x_{u,3}$ & + & $\cdots$ & + & \multicolumn{1}{l|}{$x_{u,8}$} & \multicolumn{1}{l|}{+} & $x_{u,9}$ & + & \multicolumn{1}{l|}{$x_{u,10}$} &  &  &  &  &  &  &  &  &  &  & $\leq 1$ \\ \cline{5-6} \cline{14-14}
     &  &  &  &  & \multicolumn{1}{l|}{$\ddots$} &  &  & \multicolumn{1}{l|}{$\vdots$} & \multicolumn{1}{l|}{} & $\vdots$ &  &  & \multicolumn{1}{l|}{$\ddots$} &  &  &  &  &  &  &  &  &  & $\vdots$ \\ \cline{7-8} \cline{15-17}
     &  &  &  &  &  &  & \multicolumn{1}{l|}{} & \multicolumn{1}{l|}{$x_{u,8}$} & \multicolumn{1}{l|}{+} & $x_{u,9}$ & + & $\cdots$ & + & $x_{u,14}$ & + & \multicolumn{1}{l|}{$x_{u,15}$} &  &  &  &  &  &  & $\leq 1$ \\ \cline{9-9} \cline{18-19}
     &  &  &  &  &  &  &  &  & \multicolumn{1}{l|}{} & $x_{u,9}$ & + & $\cdots$ & + & $x_{u,14}$ & + & $x_{u,15}$ & + & \multicolumn{1}{l|}{$x_{u,16}$} &  &  &  &  & $\leq 1$ \\ \cline{11-19}
     &  &  &  &  &  &  &  &  &  &  &  &  &  &  &  &  &  &  &  &  &  &  &  \\ \cline{11-19}
     &  &  &  &  &  &  &  &  & \multicolumn{1}{l|}{} & $x_{u,9}$ & + & $\cdots$ & + & $x_{u,14}$ & + & $x_{u,15}$ & + & \multicolumn{1}{l|}{$x_{u,16}$} &  &  &  &  & $\leq 1$ \\ \cline{11-12} \cline{21-21}
     &  &  &  &  &  &  &  &  &  &  & \multicolumn{1}{l|}{} & $\cdots$ & + & $x_{u,14}$ & + & $x_{u,15}$ & + & \multicolumn{1}{l|}{$x_{u,16}$} & \multicolumn{1}{l|}{+} & \multicolumn{1}{l|}{$x_{u,17}$} &  &  & $\leq 1$ \\ \cline{13-14} \cline{22-22}
     &  &  &  &  &  &  &  &  &  &  &  &  & \multicolumn{1}{l|}{$\ddots$} &  &  & $\vdots$ &  & \multicolumn{1}{l|}{$\vdots$} & \multicolumn{1}{l|}{} &  & \multicolumn{1}{l|}{$\ddots$} &  & $\vdots$ \\ \cline{15-16} \cline{21-22}
     &  &  &  &  &  &  &  &  &  &  &  &  &  &  & \multicolumn{1}{l|}{} & $x_{u,15}$ & + & \multicolumn{1}{l|}{$x_{u,16}$} &  &  &  &  & $\leq 1$ \\ \cline{17-18}
     &  &  &  &  &  &  &  &  &  &  &  &  &  &  &  &  & \multicolumn{1}{l|}{} & \multicolumn{1}{l|}{$x_{u,16}$} &  &  &  &  & $\leq 1$ \\ \cline{19-19}
    \end{tabular}
}

\vspace{0.5em}
\textit{Range variable illustration (backward chain with $w=8$):}\\[0.5em]
\begin{tabular}{lllllllllllllll}
\cline{1-11}
\multicolumn{1}{|l}{$x_{u,1}$} & + & $x_{u,2}$ & + & $x_{u,3}$ & + & $\cdots$ & + & $x_{u,7}$ & + & \multicolumn{1}{l|}{$x_{u,8}$} &  &  &  & $\Leftrightarrow R_{u,1,8}$ \\ \cline{1-2}
 & \multicolumn{1}{l|}{} & $x_{u,2}$ & + & $x_{u,3}$ & + & $\cdots$ & + & $x_{u,7}$ & + & \multicolumn{1}{l|}{$x_{u,8}$} &  &  &  & $\Leftrightarrow R_{u,2,8}$ \\ \cline{3-4}
 &  &  & \multicolumn{1}{l|}{} & $x_{u,3}$ & + & $\cdots$ & + & $x_{u,7}$ & + & \multicolumn{1}{l|}{$x_{u,8}$} &  &  &  & $\Leftrightarrow R_{u,3,8}$ \\ \cline{5-6}
 &  &  &  &  & \multicolumn{1}{l|}{$\ddots$} &  &  & $\vdots$ &  & \multicolumn{1}{l|}{$\vdots$} &  &  &  & $\vdots$ \\ \cline{7-8}
 &  &  &  &  &  &  & \multicolumn{1}{l|}{} & $x_{u,7}$ & + & \multicolumn{1}{l|}{$x_{u,8}$} &  &  &  & $\Leftrightarrow R_{u,7,8}$ \\ \cline{9-10}
 &  &  &  &  &  &  &  &  & \multicolumn{1}{l|}{} & \multicolumn{1}{l|}{$x_{u,8}$} &  &  &  & $\Leftrightarrow R_{u,8,8}$ \\ \cline{11-11}
\end{tabular}

\vspace{1em}
\textit{Range variable illustration (forward chain for last block):}\\[0.5em]
\begin{tabular}{llllllllll}
\cline{1-1}
\multicolumn{1}{|l|}{$x_{u,8(m-1)+1}$} &  &  &  &  &  &  &  &  & $\Leftrightarrow R_{u,8(m-1)+1,8(m-1)+1}$ \\ \cline{2-3}
\multicolumn{1}{|l}{$x_{u,8(m-1)+1}$} & + & \multicolumn{1}{l|}{$x_{u,8(m-1)+2}$} &  &  &  &  &  &  & $\Leftrightarrow R_{u,8(m-1)+1,8(m-1)+2}$ \\ \cline{4-5}
\multicolumn{1}{|l}{$x_{u,8(m-1)+1}$} & + & $x_{u,8(m-1)+2}$ & + & \multicolumn{1}{l|}{$\ddots$} &  &  &  &  & $\vdots$ \\ \cline{6-7}
\multicolumn{1}{|l}{$x_{u,8(m-1)+1}$} & + & $x_{u,8(m-1)+2}$ & + & $\cdots$ & + & \multicolumn{1}{l|}{$x_{u,k}$} &  &  & $\Leftrightarrow R_{u,8(m-1)+1,k}$ \\
\cline{1-7}
\end{tabular}
\caption{Staircase block structure for a vertex $u$ with block width $w=8$. The backward chain (top-left) builds ranges starting from $x_{u,1}$ and extending toward the block boundary. The bidirectional chains (middle) support both directions for inter-block ranges. The forward chain (bottom-right) builds ranges ending at the block boundary.}
\label{fig:block-structure}
\end{figure*}

\subsubsection{Distance Constraints}

The distance constraint $|c(u) - c(v)| \geq d_{u,v}$ is equivalent to preventing any pair of colors $(i, j)$ such that $|i-j| < d_{u,v}$. This can be formulated as a cardinality constraint:
\begin{equation}\label{eq:block-dist}
\sum_{i=c}^{c+d_{u,v}-1} x_{u,i} + \sum_{i=c}^{c+d_{u,v}-1} x_{v,i} \leq 1
\end{equation}
for all valid offsets $c \in [1, k - d_{u,v} + 1]$. The term $\sum_{i=c}^{c+d_{u,v}-1} x_{u,i}$ represents whether vertex $u$ is assigned a color in the conflicting window $[c, c+d_{u,v}-1]$.

\subsubsection{Expressing Range Sums Using Block Variables}

The key insight of the block encoding is that the range sum $\sum_{i=c}^{c+d_{u,v}-1} x_{u,i}$ can be expressed using the auxiliary $R$ variables, avoiding the need to enumerate all individual assignment variables. Let $B$ denote the block containing position $c$, with start and end positions $s(B)$ and $e(B)$ respectively.

\textit{Case 1: Range aligns with block boundaries.}
If $c = s(B)$ and $c + d_{u,v} - 1 = e(B)$, then:
\begin{equation}
\sum_{i=c}^{c+d_{u,v}-1} x_{u,i} = R_{u, s(B), e(B)}
\end{equation}

\textit{Case 2: Range spans two adjacent blocks.}
If the range $[c, c+d_{u,v}-1]$ crosses a block boundary into an adjacent block $B'$, then:
\begin{equation}
\sum_{i=c}^{c+d_{u,v}-1} x_{u,i} = R_{u, c, e(B)} + R_{u, s(B'), c+d_{u,v}-1}
\end{equation}

\textit{Case 3: Range is strictly contained within a block.}
This is the critical case that motivates the subtraction operation. When the range $[c, c+d_{u,v}-1]$ lies strictly within a block (i.e., $s(B) < c$ and $e(B) > c+d_{u,v}-1$), a subtraction is required:
\begin{itemize}
    \item For a backward-chained block (first block):
    \begin{equation}
    \sum_{i=c}^{c+d_{u,v}-1} x_{u,i} = R_{u, c, e(B)} - R_{u, c+d_{u,v}, e(B)}
    \end{equation}
    \item For a forward-chained block (last block):
    \begin{equation}
    \sum_{i=c}^{c+d_{u,v}-1} x_{u,i} = R_{u, s(B), c+d_{u,v}-1} - R_{u, s(B), c-1}
    \end{equation}
    \item For bidirectional blocks (middle blocks), either direction can be used.
\end{itemize}

\textit{Motivation for subtraction operations.}
The subtraction formulation is essential because it allows using a uniform block width (e.g., $w = \max_{\{u,v\} \in E} d_{u,v}$) for all vertices, regardless of the specific edge weights incident to each vertex. Without subtraction, one would need to construct blocks with different widths for different edge weights, significantly complicating the encoding structure and increasing the number of auxiliary variables. With subtraction, a single staircase block structure suffices to handle all edge weights up to the block width.

\subsubsection{Encoding Subtraction Operations}

We present two approaches for encoding the subtraction $A - B$ that appears in the range sum expressions.

\textit{1. Direct encoding without auxiliary variables (X).}
The X encoding substitutes the subtraction directly into the pairwise constraint without introducing auxiliary variables. For the constraint $(A - B) + C \leq 1$ (where $A, B, C$ are $R$ variables), the logical interpretation is: ``if $A$ is true and $B$ is false, then $C$ must be false.'' This translates to the clause:
\begin{equation}
\neg A \vee B \vee \neg C
\end{equation}

More generally, for the four possible cases in Eq.~\eqref{eq:block-dist}:
\begin{itemize}
    \item If both terms are simple: $R_{u,a,b} + R_{v,c,d} \leq 1$ becomes $\neg R_{u,a,b} \vee \neg R_{v,c,d}$.
    \item If one term involves subtraction: $(R_{u,a,b} - R_{u,a',b'}) + R_{v,c,d} \leq 1$ becomes $\neg R_{u,a,b} \vee R_{u,a',b'} \vee \neg R_{v,c,d}$.
    \item If both terms involve subtraction: $(R_{u,a,b} - R_{u,a',b'}) + (R_{v,c,d} - R_{v,c',d'}) \leq 1$ becomes $\neg R_{u,a,b} \vee R_{u,a',b'} \vee \neg R_{v,c,d} \vee R_{v,c',d'}.$
\end{itemize}

\textit{2. Encoding with auxiliary variables (Xa).}
The Xa encoding introduces an auxiliary variable $S$ to represent the subtraction $S \Leftrightarrow (R_1 \wedge \neg R_2)$, encoded as:
\begin{align}
R_1 \vee \neg R_2 \vee S \\
\neg R_1 \vee R_2 \vee S \\
\neg S \vee R_1 \\
\neg S \vee \neg R_2
\end{align}

Since the same subtraction term (e.g., $R_{u, s, e} - R_{u, s', e'}$) often recurs across different distance constraints, the Xa encoding maintains a map of generated subtraction variables and reuses them when the same term is encountered. This significantly reduces the number of auxiliary variables.

\subsection{Symmetry Breaking}

Graph coloring problems exhibit color symmetry: permuting color labels in a valid solution produces another valid solution~\citep{vangelder2008,mendezdiaz2008}. To break this symmetry, we employ a reflection-based technique that restricts a selected vertex to the lower half of the color range.
The key observation is that for any valid coloring $c$ with span $k$, the reflected coloring $c'(v) = k + 1 - c(v)$ is also valid with the same span. This reflection property arises because the distance constraint $|c(u) - c(v)| \geq d_{u,v}$ is preserved under the transformation:
\begin{equation}
|c'(u) - c'(v)| = |(k + 1 - c(u)) - (k + 1 - c(v))| = |c(u) - c(v)| \geq d_{u,v}
\end{equation}

Based on this observation, we fix the color of the vertex with the highest degree, denoted $h$, to the lower half of the available colors. The rationale for selecting the highest-degree vertex is that it participates in the most constraints, so restricting its color range has the greatest impact on pruning the search space. The validity of this symmetry breaking technique is established in Proposition~\ref{prop:symmetry} (Appendix~\ref{appendix:symmetry}). The specific constraints depend on the encoding type:
\begin{itemize}
    \item \textit{1G:} $y_{h, \lfloor k/2 \rfloor + 1} = \text{false}$ (forces $c(h) \leq \lfloor k/2 \rfloor$).
    \item \textit{1L:} $y_{h, \lfloor k/2 \rfloor} = \text{true}$ (forces $c(h) \leq \lfloor k/2 \rfloor$).
    \item \textit{2G / 2L / Block:} $\neg x_{h, j}$ for all $j > \lfloor k/2 \rfloor$.
\end{itemize}
\subsection{Solving Modes}

\begin{algorithm}[htb]
\caption{Optimal Solving for BCP}
\label{alg:optimal-solving}
\begin{algorithmic}[1]
\Require Graph $G$, upper bound $H$, encoding type, incremental mode flag
\Ensure Optimal span $k^*$
\State $k \gets H$
\State Encode BCP instance for span $k$
\State result $\gets$ SAT\_Solve()
\If{result $=$ UNSAT}
    \State \Return INFEASIBLE
\EndIf
\While{result $=$ SAT and $k > 1$}
    \State $k \gets k - 1$
    \If{incremental mode}
        \State Add assumptions to restrict span to $k$ \Comment{$\neg y_{u,k+1}$ or $y_{u,k}$}
        \State result $\gets$ SAT\_Solve\_With\_Assumptions()
    \Else
        \State Reset solver
        \State Encode BCP instance for span $k$
        \State result $\gets$ SAT\_Solve()
    \EndIf
\EndWhile
\State $k^* \gets k + 1$ \Comment{Last satisfiable span}
\State \Return $k^*$
\end{algorithmic}
\end{algorithm}

The framework supports two solving modes, as illustrated in Algorithm~\ref{alg:optimal-solving}. In the \textit{non-incremental mode}, the solver resets completely for each candidate span value $k$. Starting from the upper bound $H$, the solver encodes and solves the problem for span $k$. If satisfiable, the span is decreased by 1 and the process repeats with a fresh encoding.

In the \textit{incremental mode}, the solver reuses learned clauses across iterations by using SAT assumptions. For the greater-than encodings, the assumption $\neg y_{u,k}$ for all $u \in V$ effectively restricts the span to at most $k-1$ without re-encoding the entire formula. For the less-than encodings, the assumption $y_{u,k-1}$ for all $u \in V$ achieves the same effect. The block encoding does not support incremental mode due to the complex staircase block structure that depends on the span value. The correctness of the proposed encodings, namely the property that the SAT formula is satisfiable if and only if a valid BCP coloring exists, is established formally in Appendix~\ref{appendix:correctness}.


\section{Computational Experiments}\label{sec:experiments}

This section presents the experimental evaluation of the proposed framework. We compare 36 configurations of the proposed method against the state-of-the-art SAT-based approaches POP-S-B and POPH-S-B~\citep{faber2024} on standard benchmark instances. We focus our comparison on these SAT-based methods because \citet{faber2024} demonstrated that they outperform earlier approaches based on integer linear programming (ILP) and constraint programming (CP) from~\citet{dias2021}, establishing them as the current state of the art for exact BCP solving.

\subsection{Benchmark Instances}

We evaluate on two benchmark sets commonly used in the BCP literature. All benchmark instances are available from the Graph Coloring and its Generalizations archive.\footnote{\url{https://mat.gsia.cmu.edu/COLOR03/}}

\textit{GEOM set.} 33 instances from Michael Trick's Graph Coloring benchmarks~\citep{trick2002}, ranging from GEOM20 to GEOM120 with variants (a, b). These instances have $|V| \in [20, 120]$ vertices and represent geometric random graphs with varying densities.

\textit{MS-CAP set.} 18 instances derived from cellular network frequency assignment problems, including the c21 and c55 families from~\citet{dias2021}. These instances are transformed from Bandwidth MultiColoring Problem instances using the vertex-to-clique transformation described in Section~\ref{sec:problem}.

The complete implementation of the proposed framework is publicly available at \url{https://github.com/bcp-ese/bcp-ese}.

\subsection{Experimental Setup}

All experiments were conducted on a computing server equipped with an Intel Core i5-12500H processor (16 threads) and 16 GB RAM, using the CaDiCaL 1.9.5 SAT solver~\citep{cadical} with a time limit of 3600 seconds per instance.

For comparison with the state-of-the-art SAT-based methods POPH-S-B and POP-S-B from~\citet{faber2024}, we reimplemented their encodings in our framework using the same SAT solver (CaDiCaL 1.9.5) to ensure a fair comparison under identical experimental conditions.

\subsection{Overview of Results}

This subsection presents a comprehensive analysis of the experimental results across all 36 configurations and 51 benchmark instances (33 GEOM + 18 MS-CAP). Table~\ref{table:config-overview} summarizes the performance of the top configurations, including comparisons with state-of-the-art SAT-based methods.

\begin{table}[!htb]
\centering
\caption{Performance overview on the combined benchmark.}
\begin{tabularx}{\textwidth}{c c c c c c c c c}
\toprule
\textbf{Rank} & \textbf{Method} & \textbf{Width} & \textbf{Inc.} & \textbf{Sym.} & \textbf{Solved} & \textbf{Time} & \textbf{Vars} & \textbf{Cls} \\
 &  &  &  &  &  & \textbf{(s)} & $(\times 10^3)$ & $(\times 10^4)$ \\
\midrule
1 & Xa & fixed & $x$ & true & \textbf{51} & \textbf{2721.3} & 897.2 & 633.3 \\
2 & 2L & -- & -- & true & \textbf{51} & 2817.1 & 255.4 & 153.9 \\
3 & 1G & -- & -- & true & \textbf{51} & 2834.0 & \textbf{127.7} & \textbf{115.9} \\
4 & POPH-S-B$^\dagger$ & -- & -- & false & \textbf{51} & 2904.6 & 255.4 & 153.9 \\
5 & 2L & -- & -- & false & \textbf{51} & 3041.1 & 255.4 & 153.9 \\
6 & 2G & -- & -- & true & \textbf{51} & 3082.5 & 255.4 & 153.9 \\
7 & 1L & -- & -- & false & \textbf{51} & 3177.0 & \textbf{127.7} & \textbf{115.9} \\
8 & 2G & -- & -- & false & \textbf{51} & 3194.3 & 255.4 & 153.9 \\
9 & X & vary & $x$ & true & \textbf{51} & 3202.1 & 470.4 & 461.4 \\
10 & Xa & fixed & -- & false & \textbf{51} & 3205.5 & 623.5 & 426.4 \\
11 & POP-S-B$^\dagger$ & -- & -- & false & \textbf{51} & 3307.5 & \textbf{127.7} & \textbf{115.9} \\
\bottomrule
\end{tabularx}
\begin{flushleft}
\footnotesize{Inc.: incremental mode; Sym.: symmetry breaking; Vars: total number of variables; Cls: total number of clauses.\\
$^\dagger$ State-of-the-art SAT-based methods from~\citet{faber2024}.}
\end{flushleft}
\label{table:config-overview}
\end{table}

The results in Table~\ref{table:config-overview} suggest several findings regarding the proposed framework. First, the block encoding Xa with incremental solving and symmetry breaking achieves the best overall performance among the tested configurations, reducing total solving time by 6.3\% compared to POPH-S-B (2721.3s vs.\ 2904.6s) and by 17.7\% compared to POP-S-B (2721.3s vs.\ 3307.5s). Second, the top three configurations outperform both state-of-the-art methods on this benchmark set, indicating consistent improvements across different encoding approaches.

The relationship between encoding size and solving efficiency reveals a notable pattern. Although 1G and POP-S-B generate the fewest variables (127.7K) and clauses (115.9K), their solving times differ by 16.7\% (2834.0s vs.\ 3307.5s). This gap suggests that encoding compactness alone does not determine performance, a phenomenon well-documented in SAT encoding research~\citep{vangelder2008,ansotegui2004}. The semantic structure of clauses and their interaction with solver propagation mechanisms appear to play an equally important role.

Conversely, the Xa encoding generates seven times more variables and 5.5 times more clauses than the most compact encodings, yet achieves the fastest solving time in our experiments. This observation is consistent with the principle that auxiliary variables can strengthen unit propagation during conflict analysis~\citep{sinz2005}; the additional variables may create ``shortcuts'' that help the solver derive implications more efficiently.

The role of symmetry breaking appears to be an important factor. Comparing 2L with and without symmetry breaking (ranks 2 vs.\ 5), symmetry breaking provides a 7.4\% speedup. However, the effectiveness varies by encoding, consistent with findings that symmetry breaking interacts with encoding structure in complex ways~\citep{hebrard2020}.

Notably, the order-based encodings (2L, 1G) achieve solving times within 4\% of the best configuration despite their simpler formulation, offering practitioners a favorable trade-off between implementation complexity and computational efficiency~\citep{faber2024}.

\subsection{Cumulative Success Analysis}

\begin{figure}[!htb]
\centering
\begin{tikzpicture}
\begin{axis}[
    width=0.95\textwidth,
    height=6.5cm,
    xlabel={Time limit (seconds)},
    ylabel={Number of instances solved},
    xmin=0, xmax=2700,
    ymin=47.5, ymax=51.5,
    xtick={0, 500, 1000, 1500, 2000, 2500},
    ytick={48, 49, 50, 51},
    legend pos=south east,
    legend style={font=\footnotesize},
    grid=major,
    grid style={dashed, gray!30},
]
\addplot[thick, solid, mark=*, mark size=2pt, blue] coordinates {
    (500, 49) (1000, 50) (1500, 51) (2000, 51) (2500, 51)
};
\addplot[thick, dashdotted, mark=square*, mark size=2pt, cyan] coordinates {
    (500, 49) (1000, 49) (1500, 51) (2000, 51) (2500, 51)
};
\addplot[thick, dotted, mark=diamond*, mark size=2pt, blue!60] coordinates {
    (500, 48) (1000, 50) (1500, 51) (2000, 51) (2500, 51)
};
\addplot[thick, dashed, mark=triangle*, mark size=2pt, red] coordinates {
    (500, 50) (1000, 50) (1500, 50) (2000, 50) (2500, 51)
};
\legend{Xa$^*$, X$^*$, Xa, Order-based / SOTA}
\end{axis}
\end{tikzpicture}
\caption{Cumulative success curves showing instances solved as a function of time limit. Xa$^*$ = Xa (fixed, $x$, Sym.); X$^*$ = X (vary, $x$, Sym.); Xa = Xa (fixed, no Inc., no Sym.).}
\label{fig:cumulative-success}
\end{figure}
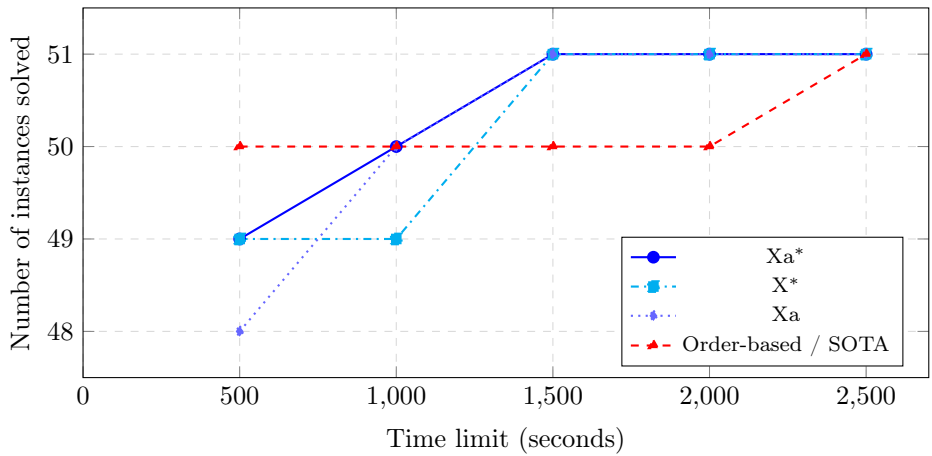

\begin{table}[!htb]
\centering
\caption{Number of instances solved at each time threshold (seconds).}
\label{table:cumulative-success}
\begin{tabular}{c c c c c c c c c}
\hline
 & & & & \multicolumn{5}{c}{\textbf{Timeout threshold (seconds)}} \\
\cline{5-9}
\textbf{Method} & \textbf{Width} & \textbf{Inc.} & \textbf{Sym.} & \textbf{500} & \textbf{1000} & \textbf{1500} & \textbf{2000} & \textbf{2500--3600} \\
\hline
Xa & fixed & $x$ & true  & 49 & 50 & \textbf{51} & 51 & 51 \\
X  & vary  & $x$ & true  & 49 & 49 & \textbf{51} & 51 & 51 \\
Xa & fixed & --  & false & 48 & 50 & \textbf{51} & 51 & 51 \\
\hline
1G & -- & -- & true  & 50 & 50 & 50 & 50 & 51 \\
1L & -- & -- & false & 50 & 50 & 50 & 50 & 51 \\
2G & -- & -- & false & 50 & 50 & 50 & 50 & 51 \\
2G & -- & -- & true  & 50 & 50 & 50 & 50 & 51 \\
2L & -- & -- & false & 50 & 50 & 50 & 50 & 51 \\
2L & -- & -- & true  & 50 & 50 & 50 & 50 & 51 \\
POPH-S-B$^\dagger$ & -- & -- & false & 50 & 50 & 50 & 50 & 51 \\
POP-S-B$^\dagger$  & -- & -- & false & 50 & 50 & 50 & 50 & 51 \\
\hline
\end{tabular}
\begin{flushleft}
\footnotesize{$^\dagger$ State-of-the-art SAT-based methods from~\citet{faber2024}.}
\end{flushleft}
\end{table}

Figure~\ref{fig:cumulative-success} and Table~\ref{table:cumulative-success} present the cumulative success analysis, showing the number of instances solved as a function of the time limit. A clear dichotomy emerges between block encodings and order-based encodings in their solving trajectories.

The block encodings (Xa and X) complete all 51 instances by 1500 seconds, while all order-based encodings (1G, 1L, 2G, 2L) and the state-of-the-art methods (POPH-S-B, POP-S-B) require 2500 seconds. This difference of 1000 seconds is attributable to a single challenging instance (GEOM120b). This result suggests that block encodings may provide substantially better performance on the hardest instances.

Among the block encodings, Xa with incremental solving and symmetry breaking (Xa, fixed, $x$, Sym.) exhibits the most balanced behavior: it solves 49 instances by 500 seconds and reaches 51 by 1500 seconds. The base Xa configuration without incremental solving or symmetry breaking (Xa, fixed) initially lags with only 48 instances at 500 seconds but catches up by 1000 seconds and completes all instances by 1500 seconds. The X encoding with vary-width strategy shows slower progress at 1000 seconds (49 vs.\ 50 for Xa configurations) but still achieves full completion by 1500 seconds.

In contrast, all order-based encodings and state-of-the-art methods exhibit identical cumulative behavior: they solve 50 instances quickly (by 500 seconds) but remain stuck at 50 until 2500 seconds. This uniform behavior suggests that these encodings may share similar search space characteristics that make GEOM120b particularly challenging. The block encodings' ability to solve this instance 1000 seconds faster suggests that their constraint structure may provide more effective pruning for instances with high edge density and large span requirements.

\subsection{Performance on Hard Instances}

This subsection analyzes the performance on the ten most computationally demanding instances. Figure~\ref{fig:hard-detailed} presents the solving time for each configuration on each of the ten hard instances, while Figure~\ref{fig:hard-summary} summarizes GEOM120b performance and cumulative solving time.

\begin{figure}[!htb]
\centering
\begin{tikzpicture}
\begin{axis}[
    width=0.98\textwidth,
    height=8cm,
    xlabel={Instance},
    ylabel={Time (seconds)},
    ymin=0, ymax=2500,
    symbolic x coords={GEOM80a, GEOM80b, GEOM90a, GEOM90b, GEOM100a, GEOM100b, GEOM110a, GEOM110b, GEOM120a, GEOM120b},
    xtick=data,
    x tick label style={rotate=45, anchor=east, font=\footnotesize},
    legend pos=north west,
    legend style={font=\scriptsize, cells={anchor=west}},
    legend columns=2,
    grid=major,
    grid style={dashed, gray!30},
    ytick={0, 500, 1000, 1500, 2000, 2500},
]
\addplot[thick, solid, mark=*, mark size=2pt, blue] coordinates {
    (GEOM80a, 5.99) (GEOM80b, 6.05) (GEOM90a, 9.02) (GEOM90b, 56.45)
    (GEOM100a, 94.09) (GEOM100b, 152.63) (GEOM110a, 233.17) (GEOM110b, 714.07)
    (GEOM120a, 425.40) (GEOM120b, 1011.60)
};
\addplot[thick, dotted, mark=triangle*, mark size=2pt, green!60!black] coordinates {
    (GEOM80a, 3.81) (GEOM80b, 6.79) (GEOM90a, 7.99) (GEOM90b, 25.08)
    (GEOM100a, 252.86) (GEOM100b, 96.71) (GEOM110a, 99.08) (GEOM110b, 183.25)
    (GEOM120a, 111.07) (GEOM120b, 2015.05)
};
\addplot[thick, dashdotted, mark=diamond*, mark size=2pt, orange] coordinates {
    (GEOM80a, 2.77) (GEOM80b, 5.31) (GEOM90a, 7.27) (GEOM90b, 36.42)
    (GEOM100a, 249.19) (GEOM100b, 63.72) (GEOM110a, 112.03) (GEOM110b, 195.44)
    (GEOM120a, 128.83) (GEOM120b, 2019.91)
};
\addplot[thick, solid, mark=pentagon*, mark size=2pt, red] coordinates {
    (GEOM80a, 5.56) (GEOM80b, 5.69) (GEOM90a, 7.02) (GEOM90b, 26.85)
    (GEOM100a, 225.83) (GEOM100b, 61.44) (GEOM110a, 222.96) (GEOM110b, 230.09)
    (GEOM120a, 103.53) (GEOM120b, 2004.59)
};
\addplot[thick, dashed, mark=star, mark size=2.5pt, purple] coordinates {
    (GEOM80a, 4.58) (GEOM80b, 4.19) (GEOM90a, 5.52) (GEOM90b, 22.62)
    (GEOM100a, 204.91) (GEOM100b, 83.19) (GEOM110a, 295.25) (GEOM110b, 228.58)
    (GEOM120a, 106.48) (GEOM120b, 2342.27)
};
\legend{Xa$^*$, 2L$^*$, 1G$^*$, POPH-S-B, POP-S-B}
\end{axis}
\end{tikzpicture}
\caption{Solving time for each configuration on the ten hardest instances. Xa$^*$ = Xa (fixed, $x$, Sym.); 2L$^*$ = 2L (Sym.); 1G$^*$ = 1G (Sym.).}
\label{fig:hard-detailed}
\end{figure}
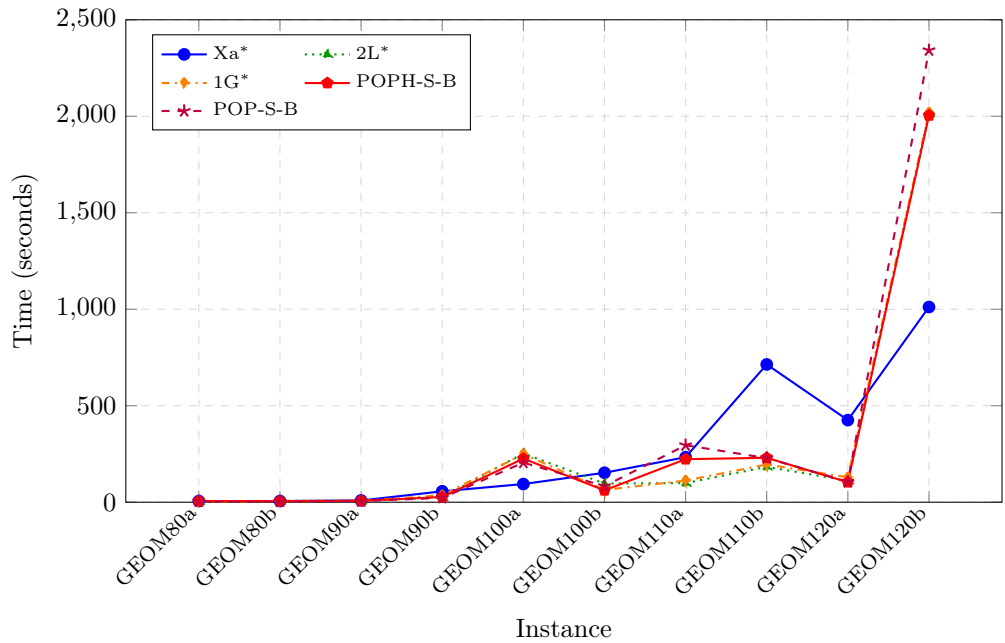

\begin{figure}[!htb]
\centering
\begin{tikzpicture}
\begin{axis}[
    width=0.95\textwidth,
    height=7cm,
    xlabel={Configuration},
    ylabel={Time (seconds)},
    ymin=0, ymax=3600,
    symbolic x coords={Xa-star, 2L-star, 1G-star, POPH-S-B, Xa, X-star, POP-S-B},
    xtick=data,
    xticklabels={Xa$^*$, 2L$^*$, 1G$^*$, POPH-S-B, Xa, X$^*$, POP-S-B},
    x tick label style={font=\footnotesize},
    legend pos=north west,
    legend style={font=\footnotesize},
    grid=major,
    grid style={dashed, gray!30},
    ytick={0, 500, 1000, 1500, 2000, 2500, 3000, 3500},
]
\addplot[thick, solid, mark=*, mark size=3pt, blue] coordinates {
    (Xa-star, 1011.6)
    (2L-star, 2015.05)
    (1G-star, 2019.91)
    (POPH-S-B, 2004.59)
    (Xa, 851.36)
    (X-star, 1363.97)
    (POP-S-B, 2342.27)
};
\addplot[thick, dashed, mark=triangle*, mark size=3pt, red] coordinates {
    (Xa-star, 2708.46)
    (2L-star, 2801.67)
    (1G-star, 2820.90)
    (POPH-S-B, 2893.56)
    (Xa, 3181.60)
    (X-star, 3190.39)
    (POP-S-B, 3297.58)
};
\legend{GEOM120b, Top 10 instances}
\end{axis}
\end{tikzpicture}
\caption{GEOM120b solving time and cumulative time for the ten hardest instances. Xa$^*$ = Xa (fixed, $x$, Sym.); X$^*$ = X (vary, $x$, Sym.); Xa = Xa (fixed, no Inc., no Sym.); 2L$^*$ = 2L (Sym.); 1G$^*$ = 1G (Sym.).}
\label{fig:hard-summary}
\end{figure}
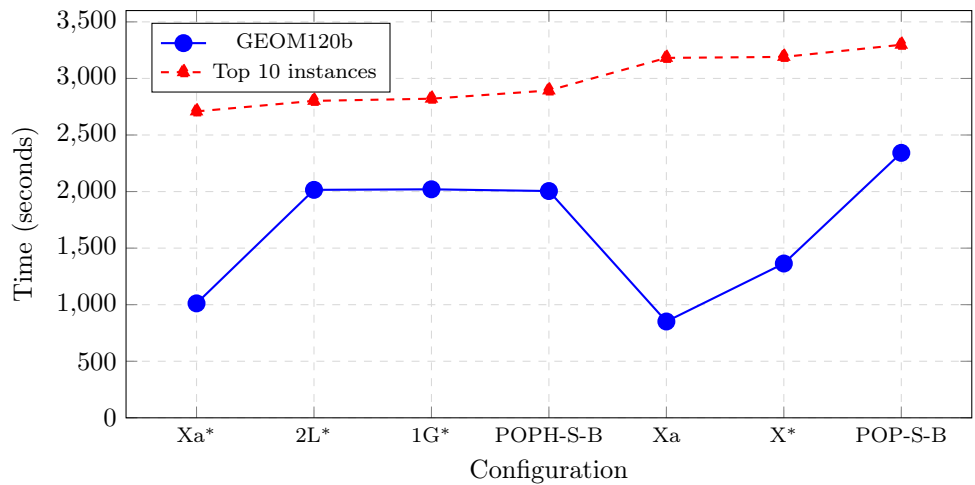

Figure~\ref{fig:hard-detailed} reveals that the block encodings exhibit instance-dependent behavior: while they dominate on GEOM120b and GEOM110b, order-based encodings (1G, 2L with symmetry breaking) achieve faster times on several smaller instances such as GEOM80a, GEOM90b, and GEOM110a. This suggests that the additional overhead of the block encoding's auxiliary variables is amortized only on instances requiring substantial computation time.

The results in Figure~\ref{fig:hard-summary} reveal important performance characteristics on GEOM120b, the most challenging instance. All block encodings (Xa and X variants) solve GEOM120b substantially faster than order-based methods. Notably, the base Xa configuration without incremental solving or symmetry breaking achieves the fastest GEOM120b time (851.4s), followed by Xa with full optimization (1011.6s) and X with vary-width (1364.0s). All three block configurations solve GEOM120b approximately twice as fast as the state-of-the-art POPH-S-B (2004.6s) and POP-S-B (2342.3s).

However, cumulative performance across the ten hardest instances tells a different story. The fully-optimized Xa (fixed, $x$, Sym.) achieves the best cumulative time (2708.5s), outperforming 2L (Sym.) by 3.3\%, 1G (Sym.) by 4.0\%, and POPH-S-B by 6.4\%. The base Xa configuration, despite its superior GEOM120b performance, ranks lower in cumulative time (3181.6s) due to slower performance on other instances. This suggests that incremental solving and symmetry breaking may provide consistent benefits across the full instance set, even when they slightly increase solving time on the single hardest instance.

The order-based encodings with symmetry breaking (2L, 1G) achieve cumulative times competitive with POPH-S-B, offering reductions of 3.2\%--3.8\%. These results suggest that the proposed encoding strategies may provide benefits across challenging instances, with block encodings performing well on the hardest instances and optimized configurations providing favorable overall performance.

\subsection{Feature Effectiveness Analysis}

This subsection quantifies the impact of each configurable feature: incremental solving, symmetry breaking, and block width strategies.

\subsubsection{Impact of Incremental Solving}

Figure~\ref{fig:incremental-analysis} compares solving times across incremental modes for each encoding type. To isolate the effect of incremental solving from the variability introduced by the hardest instance (GEOM120b), this analysis is conducted on the 50 instances that all configurations solve within the 3600-second timeout.

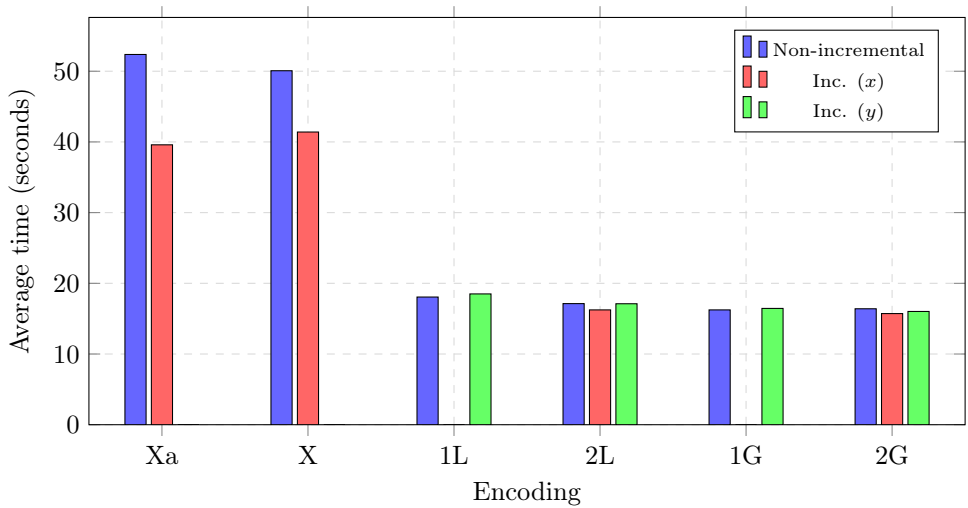
\begin{figure}[!htb]
\centering
\begin{tikzpicture}
\begin{axis}[
    width=0.98\textwidth,
    height=7cm,
    ybar,
    bar width=8pt,
    xlabel={Encoding},
    ylabel={Average time (seconds)},
    ymin=0,
    symbolic x coords={Xa, X, 1L, 2L, 1G, 2G},
    xtick=data,
    legend pos=north east,
    legend style={font=\scriptsize},
    grid=major,
    grid style={dashed, gray!30},
    nodes near coords style={font=\tiny, rotate=90, anchor=west},
]
\addplot[fill=blue!60] coordinates {
    (Xa, 52.37) (X, 50.07) (1L, 18.06) (2L, 17.13) (1G, 16.24) (2G, 16.40)
};
\addplot[fill=red!60] coordinates {
    (Xa, 39.59) (X, 41.40) (1L, 0) (2L, 16.24) (1G, 0) (2G, 15.72)
};
\addplot[fill=green!60] coordinates {
    (Xa, 0) (X, 0) (1L, 18.50) (2L, 17.11) (1G, 16.45) (2G, 16.02)
};
\legend{Non-incremental, Inc. ($x$), Inc. ($y$)}
\end{axis}
\end{tikzpicture}
\caption{Comparison of incremental solving modes across encoding types (50 instances).}
\label{fig:incremental-analysis}
\end{figure}

The effectiveness of incremental solving varies significantly across encoding types, consistent with prior observations that incremental SAT solving benefits depend on how much constraint structure is preserved between iterations~\citep{glorian2019}. For the block encodings Xa and X, incremental mode on $x$ variables provides substantial speedups: Xa (fixed, $x$) achieves a 24.4\% reduction compared to Xa (fixed), while X (vary, $x$) achieves a 17.3\% reduction compared to X (vary). When the span decreases from $k$ to $k-1$, the staircase block structure remains largely unchanged, allowing learned clauses to transfer effectively between iterations.

For the two-variable order encodings 2L and 2G, incremental solving provides moderate improvements. The $x$-variable mode achieves the best performance: 2L ($x$) provides a 5.2\% speedup over 2L, and 2G ($x$) provides a 4.1\% speedup over 2G. This suggests that learned clauses involving assignment variables transfer more effectively than those involving ordering variables.

In contrast, one-variable order encodings 1L and 1G show no benefit from incremental solving. When the span decreases, all ordering constraints are redefined relative to the new upper bound, invalidating most previously learned clauses. This is a known limitation of incremental approaches when problem structure changes significantly~\citep{glorian2019}.

\subsubsection{Impact of Symmetry Breaking}

Figure~\ref{fig:symmetry-analysis} compares solving times with and without symmetry breaking for each encoding type, including the state-of-the-art methods POPH-S-B and POP-S-B. This evaluation is also conducted on the 50 instances solved by all configurations within the timeout.

\begin{figure}[!htb]
\centering
\begin{tikzpicture}
\begin{axis}[
    width=0.98\textwidth,
    height=7cm,
    ybar,
    bar width=10pt,
    xlabel={Encoding},
    ylabel={Average time (seconds)},
    ymin=0,
    symbolic x coords={Xa, X, POPH-S-B, 1L, POP-S-B, 2L, 1G, 2G},
    xtick=data,
    x tick label style={rotate=30, anchor=east, font=\footnotesize},
    legend pos=north east,
    legend style={font=\scriptsize},
    grid=major,
    grid style={dashed, gray!30},
]
\addplot[fill=blue!60] coordinates {
    (Xa, 45.31) (X, 48.61) (POPH-S-B, 18.00) (1L, 19.57) (POP-S-B, 19.30) (2L, 16.90) (1G, 16.45) (2G, 15.15)
};
\addplot[fill=red!60] coordinates {
    (Xa, 46.66) (X, 42.86) (POPH-S-B, 20.05) (1L, 16.99) (POP-S-B, 14.58) (2L, 16.75) (1G, 16.24) (2G, 16.95)
};
\legend{SB=false, SB=true}
\end{axis}
\end{tikzpicture}
\caption{Comparison of configurations with and without symmetry breaking (50 instances).}
\label{fig:symmetry-analysis}
\end{figure}
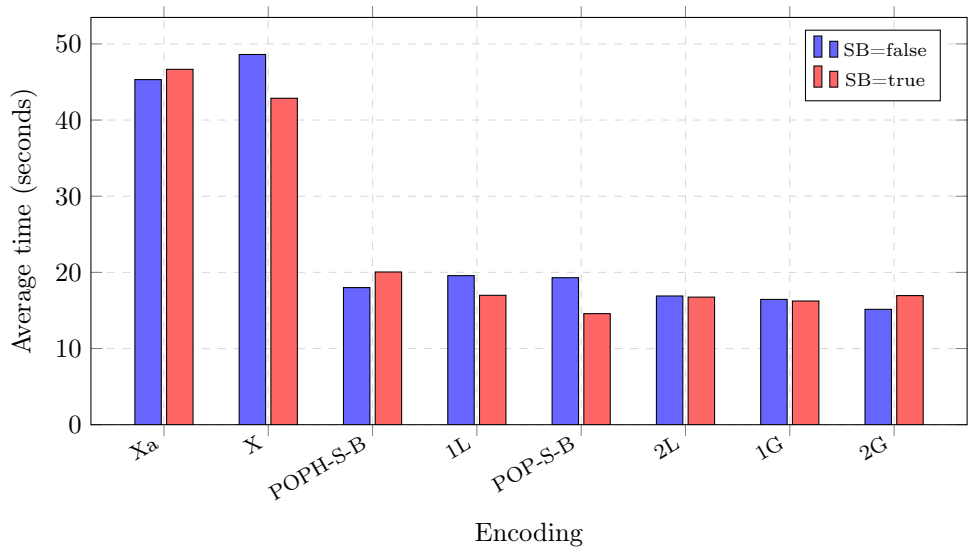

A notable observation is that all three best-performing configurations in Table~\ref{table:config-overview}---Xa (fixed, $x$, Sym.), 2L (Sym.), and 1G (Sym.)---have symmetry breaking enabled. This suggests that symmetry breaking may contribute positively to overall performance when combined with appropriate encoding choices.

However, the effectiveness of symmetry breaking exhibits strong encoding dependence, a phenomenon observed in SAT-based graph coloring~\citep{hebrard2020,mendezdiaz2008}. A notable benefit is observed for POP-S-B (Sym.), where symmetry breaking provides a 24.5\% speedup compared to POP-S-B. The configuration 1L (Sym.) shows a 13.2\% speedup over 1L, and X (vary, Sym.) achieves an 11.8\% improvement over X (vary).

In contrast, symmetry breaking degrades performance for other encodings. The configuration 2G (Sym.) exhibits an 11.9\% slowdown compared to 2G, and POPH-S-B (Sym.) experiences an 11.4\% slowdown compared to POPH-S-B. This result is consistent with findings by \citet{hebrard2020} that symmetry breaking constraints can interact unpredictably with CDCL solver heuristics, sometimes causing increased backtracking.

The divergent behavior between structurally similar encodings (POP-S-B vs.\ POPH-S-B, and Xa vs.\ X) suggests that symmetry breaking should be applied selectively based on the encoding type~\citep{mendezdiaz2008}. Nevertheless, when combined with well-suited encodings such as Xa, 2L, and 1G, symmetry breaking appears to be an important component of the best-performing configurations.

\subsubsection{Impact of Block Width Strategies}

For the block encodings (X and Xa), we analyze the impact of block width strategies. Figure~\ref{fig:width-analysis} compares configurations across fixed and varying width strategies. To isolate the effect of block width from the variability introduced by the hardest instance (GEOM120b), this analysis is conducted on the 50 instances that all configurations solve within the 3600-second timeout.

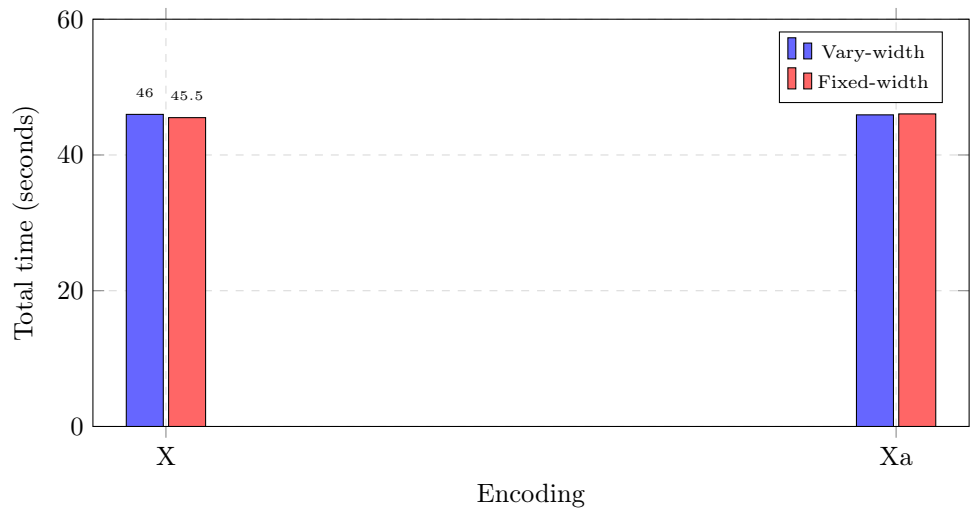
\begin{figure}[!htb]
\centering
\begin{tikzpicture}
\begin{axis}[
    width=0.98\textwidth,
    height=7cm,
    ybar,
    bar width=14pt,
    xlabel={Encoding},
    ylabel={Total time (seconds)},
    ymin=0,
    ymax=60,
    symbolic x coords={X, Xa},
    xtick=data,
    legend pos=north east,
    legend style={font=\scriptsize},
    grid=major,
    grid style={dashed, gray!30},
    nodes near coords,
    nodes near coords style={font=\tiny, /pgf/number format/fixed, /pgf/number format/precision=1},
    every node near coord/.append style={rotate=0, anchor=south, yshift=3pt},
]
\addplot[fill=blue!60] coordinates {
    (X, 45.98) (Xa, 45.91)
};
\addplot[fill=red!60] coordinates {
    (X, 45.50) (Xa, 46.05)
};
\legend{Vary-width, Fixed-width}
\end{axis}
\end{tikzpicture}
\caption{Comparison of block width strategies for block encodings (50 instances).}
\label{fig:width-analysis}
\end{figure}

The results reveal that block width strategies have minimal impact on solving performance across the 50 standard instances. For the Xa encoding, Xa (vary) achieves a marginal 0.3\% speedup (from 46.05s to 45.91s) compared to Xa (fixed). For the X encoding, X (fixed) achieves a slight 1.0\% speedup (from 45.98s to 45.50s) compared to X (vary).

These small differences suggest that block width selection is a secondary factor compared to other encoding features such as incremental solving and symmetry breaking. The similar performance across width strategies indicates that the SAT solver's clause learning mechanisms can adapt effectively to both fixed and varying block structures on instances of moderate difficulty.

The overall best configuration within this analysis is X (fixed) at 45.50s, followed closely by Xa (vary) at 45.91s, X (vary) at 45.98s, and Xa (fixed) at 46.05s. However, given the small performance differences (less than 1.2\% between best and worst), practitioners can choose either width strategy based on implementation convenience without significant performance penalty on typical instances.

\subsection{Discussion}

The experimental results present several observations about the design of SAT encodings for the BCP.

\textit{Interaction effects between encoding features.} Encoding features exhibit significant interaction effects rather than operating independently. Incremental solving benefits block encodings substantially but provides marginal or negative effects for one-variable encodings. This is consistent with the findings of \citet{glorian2019}, who observed that incremental SAT solving effectiveness depends heavily on how much constraint structure is preserved between iterations. Similarly, symmetry breaking shows strong encoding dependence, benefiting some encodings while degrading others. This phenomenon appears consistent with observations by \citet{hebrard2020} that symmetry breaking constraints can interact unpredictably with solver heuristics.

\textit{Encoding size versus propagation strength.} The Xa encoding generates significantly more variables and clauses than order-based encodings yet achieves the fastest solving time in our experiments. This observation is consistent with the well-established principle in SAT encoding design that auxiliary variables can strengthen unit propagation~\citep{sinz2005,ansotegui2004}. As noted by \citet{vangelder2008}, compact encodings do not necessarily lead to faster solving; the semantic structure of constraints and their interaction with CDCL propagation mechanisms may outweigh encoding size effects.

\textit{Instance-dependent performance.} Block encodings perform well on hard instances but show higher overhead on easier ones, where order-based encodings offer competitive performance with simpler implementation. This trade-off between encoding complexity and instance difficulty has been observed in related graph coloring work~\citep{hebrard2019,heule2022}, suggesting that adaptive encoding selection based on instance characteristics may be a worthwhile direction for future research.


\section{Conclusions}\label{sec:conclusions}

This paper presented a systematic design study of SAT encodings for the Bandwidth Coloring Problem. The proposed framework organizes encoding methods into three categories (one-variable, two-variable, and block encodings) and investigates the impact of configurable features including incremental solving, symmetry breaking, and block width strategies.

The experimental evaluation on 51 benchmark instances indicates that the proposed framework achieves competitive and, in several cases, state-of-the-art performance. The best-performing configuration in our experiments, a block encoding with incremental solving and symmetry breaking, outperforms previously published SAT-based methods by a notable margin in total solving time. On GEOM120b, the most challenging instance in our benchmark set, block encodings show speedup factors of approximately two compared to the state-of-the-art methods. Notably, several order-based encodings also outperform the state-of-the-art while maintaining simpler implementations.

An important observation from the experiments is that encoding features exhibit significant interaction effects rather than operating independently. Incremental solving benefits block encodings substantially but provides marginal or negative effects for one-variable encodings. Symmetry breaking effectiveness varies considerably across encoding types, improving some encodings while degrading others. In contrast, block width strategies show minimal impact on performance. These findings highlight the value of systematic configuration exploration and suggest that practitioners should carefully evaluate feature combinations rather than assuming universal benefit from individual optimizations.

Several directions for future research emerge from this work. First, extending the proposed framework to the Bandwidth MultiColoring Problem would allow direct handling of instances where vertices require multiple colors. Second, developing adaptive configuration selection mechanisms based on instance characteristics could help identify suitable encodings automatically. Third, exploring parallel SAT solving techniques may enable tackling larger instances. Finally, investigating hybrid approaches that combine SAT encodings with constraint propagation techniques from constraint programming solvers could yield further improvements.

\bigskip
\noindent
{\small{\bf DATA AVAILABILITY.}
The source code and experimental data are available at \url{https://github.com/bcp-ese/bcp-ese}.
}

\bigskip
\noindent
{\small{\bf FUNDING.}
The authors received no financial support for the research and the publication of this article.
}

\bigskip
\noindent
{\small{\bf AUTHORS' CONTRIBUTIONS.}
\textbf{Duc Trung Kim Nguyen}: Conceptualization, Methodology, Software, Validation, Investigation, Data Curation.
\textbf{Tuyen Van Kieu}: Validation, Writing -- Original Draft, Visualization, Writing -- Review \& Editing, Supervision.
\textbf{Khanh Van To}: Conceptualization, Methodology, Writing -- Review \& Editing, Supervision, Project Administration.
All authors have read and approved the final manuscript.
}

\bigskip
\noindent
{\small{\bf CONFLICTS OF INTEREST.}
The authors declare that they have no known competing financial interests or personal relationships that could have appeared to influence the work reported in this paper.
}

\appendix

\section{Proof of Encoding Correctness}\label{appendix:correctness}

This appendix establishes the correctness of the proposed order-based encodings.

\begin{theorem}[Correctness of Order-Based Encodings]\label{thm:correctness}
For any graph $G = (V, E)$ with edge weights $d: E \rightarrow \mathbb{N}$ and span bound $k$, the SAT formula $\Phi_k$ produced by the 1G, 1L, 2G, or 2L encoding is satisfiable if and only if there exists a valid BCP coloring with span at most $k$.
\end{theorem}
\begin{proof}
We prove the theorem for the 1G encoding; the proofs for 1L, 2G, and 2L follow by symmetric arguments.

(\textit{Soundness}) Let $\sigma$ be a satisfying assignment for $\Phi_k$. Define the coloring $c(u) = \min\{j \in [1,k] : y_{u,j} = \text{false}\}$, with $c(u) = k$ if $y_{u,k} = \text{true}$. 
\begin{itemize}
    \item The base constraint $y_{u,1}$ ensures $c(u) \geq 1$.
    \item The ordering constraints $y_{u,j} \rightarrow y_{u,j-1}$ ensure the set $\{j : y_{u,j} = \text{true}\}$ forms a prefix $[1, c(u)-1]$.
    \item The distance constraints ensure that if $c(u) = j$, then $c(v) \notin [j - d_{u,v} + 1, j + d_{u,v} - 1]$, i.e., $|c(u) - c(v)| \geq d_{u,v}$.
\end{itemize}
Thus, $c$ is a valid BCP coloring with span at most $k$.

(\textit{Completeness}) Let $c: V \rightarrow [1, k]$ be a valid BCP coloring. Define $y_{u,j} = \text{true}$ iff $c(u) \geq j$. All constraints are satisfied by construction: base constraints hold since $c(u) \geq 1$; ordering constraints hold by the prefix structure; distance constraints hold because $|c(u) - c(v)| \geq d_{u,v}$ implies $c(v)$ is outside the forbidden range.
\end{proof}

\section{Proof of Symmetry Breaking Validity}\label{appendix:symmetry}

This appendix establishes the validity of the symmetry breaking technique described in Section~\ref{sec:method}.

\begin{proposition}[Symmetry Breaking Validity]\label{prop:symmetry}
For any optimal BCP solution with span $k^*$, there exists an equivalent solution where the highest-degree vertex $h$ satisfies $c(h) \leq \lfloor k^*/2 \rfloor$.
\end{proposition}
\begin{proof}
Given any valid coloring $c$ with span $k^*$ where $c(h) > \lfloor k^*/2 \rfloor$, define the reflected coloring $c'(v) = k^* + 1 - c(v)$ for all $v \in V$. This transformation preserves:
\begin{itemize}
    \item \textit{Feasibility}: For any edge $\{u, v\}$, $|c'(u) - c'(v)| = |c(u) - c(v)| \geq d_{u,v}$.
    \item \textit{Span}: $\max_v c'(v) = k^* + 1 - \min_v c(v) = k^* + 1 - 1 = k^*$.
    \item \textit{Bound}: $c'(h) = k^* + 1 - c(h) < k^* + 1 - \lfloor k^*/2 \rfloor \leq \lfloor k^*/2 \rfloor + 1$.
\end{itemize}
Thus, applying the reflection ensures $c'(h) \leq \lfloor k^*/2 \rfloor$ without affecting the span or feasibility.
\end{proof}

\bibliographystyle{apalike} 
\bibliography{referencias.bib} 

\end{document}